\documentclass[12pt,a4paper]{article}
\usepackage{latexsym,amsmath,amsfonts,amssymb,amsthm,amscd,mathrsfs}

\usepackage{accents}

 \usepackage[unicode,colorlinks,plainpages=false,hyperindex=true,bookmarksnumbered=true,bookmarksopen=false,pdfpagelabels]{hyperref}
 \hypersetup{urlcolor=cyan,linkcolor=blue,citecolor=red,colorlinks=true}

\newcommand{\be}{\begin{equation}}
\newcommand{\ee}{\end{equation}}
\newcommand{\bea}{\begin{eqnarray}}
\newcommand{\eea}{\end{eqnarray}}

\numberwithin{equation}{section}
\newcounter{thmcounter}
\numberwithin{thmcounter}{section}

\theoremstyle{definition}
\newtheorem*{acknowledgements}{Acknowledgements}

\newtheorem{remark}[thmcounter]{Remark}
\theoremstyle{plain}

\newtheorem{lemma}[thmcounter]{Lemma}
\newtheorem{proposition}[thmcounter]{Proposition}
\newtheorem{theorem}[thmcounter]{Theorem}

\def\1{{\boldsymbol 1}}                     %
\def\C{\mathbb{C}}                          %
\def\N{\mathbb{N}}                          %
\def\R{\mathbb{R}}                          %
\def\red{\mathrm{red}}                      %
\def\id{{\mathrm{id}}}                      %
\def\dt {\left.\frac{d}{dt}\right|_{t=0}}   %
\def\fM{\mathfrak{M}}                       %
\def\cG{{\mathcal G}}                       %
\def\half{\frac{1}{2}}                      %
\def\tr{\mathrm{tr\,}}                      %
\def\red{\mathrm{red}}                      %
\def\cE{{\mathcal E}}                       %
\def\cB{{\mathcal B}}                       %
\def\cA{{\mathcal A}}                       %
\def\gl{\mathrm{gl}}                        %
\def\GL{\mathrm{GL}}                        %

\vspace*{0.5cm}
\begin{document}

\bigskip
\begin{center}
{\Large\bf
A note on quadratic Poisson brackets on $\gl(n,\R)$
related to Toda lattices}
\end{center}

\medskip
\begin{center}
L.~Feh\'er${}^{a,b,}$\footnote{Corresponding author, e-mail: lfeher@physx.u-szeged.hu} and B. Juh\'asz${}^a$  \\

\bigskip
${}^a$Department of Theoretical Physics, University of Szeged\\
Tisza Lajos krt 84-86, H-6720 Szeged, Hungary\\

\medskip
${}^b$Department of Theoretical Physics, WIGNER RCP, RMKI\\
H-1525 Budapest, P.O.B.~49, Hungary\\
\end{center}

\medskip
\begin{abstract}
It is well known that the compatible linear and quadratic Poisson brackets
of the full symmetric and of the standard open Toda lattices are restrictions of
linear and quadratic $r$-matrix Poisson brackets on the
associative algebra $\gl(n,\R)$.
 We here show that the quadratic bracket on
$\gl(n,\R)$,
corresponding to the $r$-matrix defined by the splitting of $\gl(n,\R)$
into the direct sum of the upper triangular and orthogonal Lie subalgebras, descends by Poisson reduction
from a quadratic Poisson structure on the cotangent bundle $T^* \GL(n,\R)$.
This complements the interpretation of the linear $r$-matrix
bracket as a reduction of the canonical Poisson bracket of the cotangent bundle.
\end{abstract}

\newpage

\section{Introduction}

The goal of this brief communication is to illuminate the group theoretic  origin of a certain
quadratic $r$-matrix structure on the associative algebra $\cG:= \gl(n,\R)$.
This Poisson structure is associated with  the QU factorization
and it appeared in the theory of integrable systems \cite{LP,OR}.
Like the corresponding linear $r$-matrix bracket, it can be restricted to
the Poisson submanifolds consisting of symmetric and of tridiagonal symmetric matrices \cite{OR},
thereby producing the bi-Hamiltonian structures of the full symmetric and of the usual (open)
Toda lattices \cite{DLNT,SurB}.
It is well known (see e.g. \cite{STS2}) that the linear
$r$-matrix bracket on $\cG$ is a reduction of the canonical Poisson bracket
of the cotangent bundle of the group $G:= \GL(n,\R)$.
Our observation is that $T^*G$ carries also a quadratic Poisson bracket that
descends to the relevant quadratic bracket on $\cG$ via the same reduction
procedure which works in the linear case.
The idea arises from \cite{FeAHP,FeNlin}, where bi-Hamiltonian structures
for spin Sutherland models were obtained by reducing bi-Hamiltonian structures
 on the cotangent bundle of $\GL(n,\C)$.

 We now recall the necessary background information about linear and quadratic
 $r$-matrix Poisson brackets on $\cG$.
 This is a specialization of general results
found in \cite{LP,OR} (see also \cite{SurPLA}).
Let $R$ be a linear operator on $\cG$ that solves the modified classical Yang-Baxter
equation\footnote{For reviews on $r$-matrices and their use, one may consult, for example,  \cite{STS2,SurB}.}.
Decompose $R$ as the sum of its anti-symmetric and symmetric parts, $R_a$ and $R_s$, with
respect to the non-degenerate bilinear form,
\be
\langle X,Y \rangle  := \tr(XY), \quad \forall X,Y\in \cG,
\label{I1}\ee
and suppose that $R_a$ solves the same equation as $R$.
For a smooth real function on $\cG$ let $df$ denote its gradient defined using the trace form \eqref{I1},
and introduce the  `left- and right-derivatives' $\nabla f$ and $\nabla'f$  by
\be
\nabla f(L):= L df(L), \qquad \nabla' f(L):= df(L) L.
\label{I2}\ee
Then the following formula defines a Poisson bracket on $\cG$:
\be
\{f,h\}_2 := \langle \nabla f, R_a \nabla h \rangle -  \langle \nabla' f, R_a \nabla' h \rangle
+ \langle \nabla f, R_s \nabla' h \rangle - \langle \nabla' f, R_s \nabla h \rangle.
\label{I3}\ee
The  Lie derivative of this quadratic $r$-matrix bracket along the  vector field $V(L):= \1_n$
is the linear $r$-matrix bracket,
\be
\{ f, h \}_1(L) = \langle L, [Rdf(L),dh(L)] + [ df(L), R dh(L)]\rangle,
\label{I4}\ee
and thus the two Poisson brackets are compatible.  The Hamiltonians $h_k(L) := \frac{1}{k} \tr(L^k)$
are in involution with respect to both brackets.  They enjoy
the relation
\be
\{ f, h_k\}_2 = \{f, h_{k+1}\}_1, \qquad \forall f\in C^\infty(\cG),
\label{I5}\ee
and their Hamiltonian vector fields engender
 bi-Hamiltonian Lax equations:
\be
\partial_{t_k}(L) := \{ L, h_k\}_2 = \{ L, h_{k+1}\}_1  = [R(L^k), L],
\qquad \forall k\in \N.
\label{I6}\ee

Turning to the example of our interest, let us decompose any $X\in \cG$ as
\be
X= X_> + X_0 + X_<,
\label{I7}\ee
where $X_>$, $X_0$ and $X_<$ are the strictly upper triangular, diagonal and strictly lower triangular parts of the matrix $X$,
respectively.
Denote $\cA < \cG$ the Lie subalgebra of skew-symmetric matrices and $\cB < \cG$ the
subalgebra of upper triangular matrices.
They enter the vector space direct sum
\be
\cG = \cA + \cB,
\label{I8}\ee
and, using the projections $\pi_\cA$ onto $\cA$ and $\pi_\cB$ onto $\cB$,
yield the $r$-matrix
\be
R= \half (\pi_\cB - \pi_\cA).
\label{I9}\ee
In terms of the triangular decomposition \eqref{I7},
\be
\pi_\cA(X) = X_< - (X_<)^T, \quad
\pi_\cB(X) = X_> + X_0 + (X_<)^T,
\label{I10}\ee
and
\be
R(X) = \half (X_> + X_0- X_<)  + (X_<)^T,
\,\,
R_a(X)=\half (X_> - X_<),
\,\,
R_s(X) =\half X_0 +  (X_<)^T.
\label{I11}\ee
This $r$-matrix $R$ satisfies the conditions stipulated above, and
we are going to derive its quadratic bracket \eqref{I3} by reduction of a Poisson structure on $T^*G$.

\begin{remark}
The matrix space ${\mathrm{mat}}(n\times n, \R)$ is primarily an associative algebra,  and the notation
$\gl(n,\R)$ is usually reserved for its induced Lie algebra structure.
In this paper $\gl(n,\R)$ is understood to carry
both algebraic structures, i.e., we identify $\gl(n,\R)$ with  ${\mathrm{mat}}(n\times n,\R)$ when
using the associative product.
 This should not lead to any confusion.
\end{remark}

\section{The $r$-matrix brackets from Poisson reduction}

We start with the manifold
\be
\fM:= G \times \cG = \{ (g,L)\mid g\in G, \, L\in \cG\},
\label{T1}\ee
which is to be viewed as a model of $T^*G$ obtained via right-translations
and the identification $\cG^* \simeq \cG$ given by the trace form.
For smooth real functions $F, H\in C^\infty(\fM)$,
the following formulae define two compatible Poisson brackets:
\be
\{ F,H\}_1(g,L) =   \langle \nabla_1 F, d_2 H\rangle - \langle \nabla_1 H, d_2 F \rangle +
\langle L, [d_2 F, d_2 H]\rangle,
\label{PB1}\ee
and
\bea
&& \{ F, H\}_2(g,L)  =
\langle R_a \nabla_1 F, \nabla_1 H \rangle - \langle R_a \nabla'_1 F, \nabla'_1 H \rangle
 +\langle \nabla_2 F - \nabla_2' F, r_+\nabla_2' H  - r_- \nabla_2 H
   \rangle \nonumber \\
 && \qquad\qquad\qquad \quad
 +\langle \nabla_1 F,  r_+ \nabla_2' H  - r_-\nabla_2 H  \rangle
- \langle \nabla_1 H,  r_+ \nabla_2' F  - r_- \nabla_2 F \rangle,
\label{PB2}\eea
where
\be
r_\pm := R_a \pm \half \id.
\label{T4}\ee
The derivatives are taken at $(g,L)$, 
\be
\langle \nabla_1 F(g,L), X\rangle = \dt F(e^{tX} g, L),
\quad
\langle \nabla_1' F(g,L), X\rangle = \dt F(ge^{tX}, L), \quad \forall X\in \cG,
\label{T5}\ee
and $\nabla_2 F(g,L) := L d_2 F(g,L)$, $\nabla_2' F(g,L) := d_2F(g,L) L$ with $d_2 F$ denoting
the gradient with respect to the second argument. The first bracket is just the canonical one.
The  second one is obtained by a change of variables from the Heisenberg double \cite{STS1} of
the Poisson--Lie group $G$  equipped with the Sklyanin bracket that appears in the first two terms of \eqref{PB2}.
In the corresponding complex holomorphic case, this is explained in detail in \cite{FeAHP}.
The compatibility of the two brackets also follows by the same Lie derivative argument that works in the complex case \cite{FeAHP}.

We are interested in the  restriction of the Poisson brackets \eqref{PB1} and \eqref{PB2} to
those functions on $\fM$ that are invariant with respect to the  group
\be
S:= A \times B \quad \hbox{with}\quad A:= {\mathrm{O}}(n,\R),\,\, B:= \exp(\cB),
\label{T6}\ee
whose factors correspond to the Lie algebras $\cA$ and $\cB$ in \eqref{I8}.  That is,
$B$ consists of the upper triangular elements of $G$ having positive diagonal entries.
The action of $S$ on $\fM$ is given by letting any $(a,b)\in A \times B$
act on $(g,L)\in \fM$ by the diffeomorphism
\be
(g,L) \mapsto (a g b^{-1}, a L a^{-1}).
\label{T7}\ee
Due to the QU factorization\footnote{That is, due to the fact that the matrix multiplication $m: A \times B \to G$
is a diffeomorphism.}, every $S$ orbit in $\fM$ admits a unique representative of
the form $(\1_n,L)$. Therefore, we may associate to any smooth, $S$ invariant functions $F,H$
on $\fM$ unique smooth functions $f,h$ on $\cG$ according to the rule
\be
f(L):= F(\1_n,L), \quad h(L):= H(\1_n,L).
\label{T8}\ee
Provided that the invariant functions close under the Poisson brackets on $\fM$, we
may define the reduced Poisson brackets on $C^\infty(\cG)$ by setting
\be
\{ f,h\}_i^\red (L) := \{ F, H\}_i(\1_n,L), \qquad i=1,2.
\label{T9}\ee
In other words,
in this situation the Poisson brackets descend to the quotient space $\fM/S \simeq \cG$.
The closure is obvious for the first Poisson bracket, and for the second one we prove it below.

\begin{proposition} If $F$ and $H$ are invariant with respect to the $S$ action \eqref{T7}, then
their second Poisson bracket \eqref{PB2} takes the simplified form
\be
2\{ F, H\}_2  =
 \langle \nabla_2 F ,\nabla_2' H \rangle
 - \langle \nabla_2 H ,\nabla_2' F \rangle
 + \langle \nabla_1 F,  \nabla_2' H  +\nabla_2 H  \rangle
- \langle \nabla_1 H,   \nabla_2' F + \nabla_2 F \rangle.
\label{PB2LR}\ee
This formula implies that the Poisson bracket of two $S$ invariant functions is again $S$ invariant.
\end{proposition}
\begin{proof}
The invariance of $F$ with respect to the action of one parameter subgroups of $A$ and $B$ leads to
the conditions
\be
\langle \nabla_1' F, X \rangle = 0, \quad \forall X\in \cB
\quad \hbox{and}\quad  \langle \nabla_1 F + \nabla_2 F - \nabla_2'F, Y \rangle =0, \quad \forall Y\in \cA.
\label{T11}\ee
The first condition means that $\nabla_1' F(g,L)$ is strictly upper triangular, and since the same holds for $H$ we get
\be
\langle R_a \nabla_1' F, \nabla_1' H \rangle = 0.
\label{T12}\ee
By using the second condition in \eqref{T11}, we are going to show that the contributions
containing $R_a$ cancel from all other terms of \eqref{PB2} as well.
To do this,  it proves useful to employ the direct sum decomposition
$\cG = \cA + \cA^\perp$,
where $\cA^\perp$ consists of the symmetric matrices in $\cG$.
Accordingly,   we may decompose any element $Z\in \cG$ as
\be
Z = Z^+ + Z^-
\quad \hbox{with}\quad
Z^+\in \cA, \, Z^- \in \cA^\perp.
\label{T13}\ee
Then the second condition in \eqref{T11} means that
\be
(\nabla_1 F)^+ = (\nabla_2' F - \nabla_2 F)^+.
\label{T14}\ee
By using this together with the anti-symmetry of $R_a$ and that $R_a$ maps $\cA$ into $\cA^\perp$ and $\cA^\perp$ into $\cA$,
we derive the equalities,
\be
\langle R_a \nabla_1 F, \nabla_1 H\rangle =
\langle R_a (\nabla_1 H)^-,   (\nabla_2 F - \nabla_2' F)^+ \rangle
-\langle R_a (\nabla_1 F)^-,   (\nabla_2 H - \nabla_2' H)^+ \rangle,
\label{S1}\ee
and
\be
\langle \nabla_1 F,  R_a (\nabla_2' H  - \nabla_2 H)  \rangle
= \langle R_a (\nabla_1 F)^-,   (\nabla_2 H - \nabla_2' H)^+ \rangle
+ \langle R_a (\nabla_2' F - \nabla_2 F)^+,   (\nabla_2 H - \nabla_2' H)^- \rangle.
\label{S2}\ee
By adding up \eqref{S1} and
the terms in \eqref{S2} together with (minus one times) their counterparts having $F$ and $H$ exchanged, one
precisely cancels
$\langle \nabla_2 F - \nabla_2' F, R_a (\nabla_2' H  -  \nabla_2 H) \rangle$ in \eqref{PB2}.
Then the formula \eqref{PB2LR} results directly from \eqref{PB2}.
Having derived \eqref{PB2LR}, one sees that the right-hand side of this expression is invariant under the action \eqref{T7} of $S$.
Indeed, this is a consequence of the fact that the derivatives of invariant functions are equivariant,
meaning for example that we have
\be
\nabla_1 F (a g b^{-1}, a L a^{-1})  =a (\nabla_1 F(g,L)) a^{-1},
\quad
\nabla_2 F(a g b^{-1}, aL a^{-1})= a (\nabla_2 F(g,L)) a^{-1}.
\label{T17}\ee
This and the conjugation invariance of the trace imply the claim.
\end{proof}

The following lemma will be important below.

\begin{lemma}
The $S$ invariant function $F$ on $\fM$  and the function $f$ on $\cG$ related  by \eqref{T8} satisfy the relations
\be
\nabla_1 F(\1_n,L) = (r_+ - R_s)(\nabla' f(L) - \nabla f(L)),
\quad
d_2 F(\1_n,L) = df(L),
\label{rel**}\ee
where $R_s$ and $r_+$ are given by \eqref{I11} and \eqref{T4}.
\end{lemma}
\begin{proof}
The second relation is obvious, and it implies the identities $\nabla_2 F(\1_n,L) = \nabla f(L)$ and
$\nabla'_2 F(\1_n,L) = \nabla' f(L)$.
Since $\nabla_1 F(\1_n,L) = \nabla'_1 F(\1_n,L)$ by \eqref{T5}, we see from \eqref{T11} that
\be
\nabla_1 F(\1_n,L) = (\nabla_1 F(\1_n,L))_>,
\label{T19}\ee
where we applied the triangular decomposition \eqref{I7}.
Then, noting that the anti-symmetric part of any $X\in \cG$ is $X^+ = \frac{1}{2} ( X - X^T)$,
it follows from the equality \eqref{T14} that
\be
(\nabla_1 F(\1_n,L))_> = 2 ((\nabla_1 F(\1_n,L))^+)_> = (\nabla' f(L) - \nabla f(L))_> -  ((\nabla' f(L) - \nabla f(L))_<)^T.
\label{T20}\ee
Because  $r_+ X = X_> + \frac{1}{2} X_0$ and $R_s X = \frac{1}{2} X_0 + (X_<)^T$ by \eqref{I11} and \eqref{T4},
the statement \eqref{rel**} is obtained by combining \eqref{T19} and \eqref{T20}.
 \end{proof}

We now prove our claim about the reduction origin of the quadratic bracket
\eqref{I3}, which we could not find  in the literature.
For completeness, we also show that the linear $r$-matrix  bracket \eqref{I4} descends from \eqref{PB1},
which is a classical result \cite{STS2}.

\begin{theorem} The reductions \eqref{T9} of the  Poisson brackets \eqref{PB1} and \eqref{PB2} on the cotangent bundle
$\fM\equiv T^*\GL(n,\R) $ \eqref{T1}
defined by taking quotient by the action \eqref{T7} of the group $S$ \eqref{T6}
give the linear \eqref{I4} and quadratic \eqref{I3} $r$-matrix brackets on $\cG = \gl(n,\R)$, respectively.
\end{theorem}
\begin{proof}
We have to evaluate the expressions \eqref{T9} for $f$ and $h$ related to the $S$ invariant functions $F$ and $H$ by \eqref{T8}.
We start with the second bracket, relying on \eqref{PB2LR}.
Substitution of the relations \eqref{rel**} into \eqref{PB2LR} gives
\be
\langle \nabla_1 F, \nabla'_2 H + \nabla_2 H\rangle - \cE(F,H) = \langle r_+(\nabla' f - \nabla f) - R_s (\nabla' f - \nabla f), \nabla' h + \nabla h\rangle
- \cE(f,h),
\ee
where $\cE(F,H)$ stands for the terms obtained by  exchanging the roles of $F$ and $H$, and similarly for $f$ and $h$.
Writing $r_+ = R_a + \half \id$, we find
\be
\half \langle \nabla' f - \nabla f , \nabla' h + \nabla h\rangle - \cE(f,h) =
\langle \nabla' f ,\nabla h \rangle
 -\langle \nabla f ,\nabla' h \rangle.
 \ee
The terms containing $R_a$ and $R_s$ contribute
\be
\langle R_a(\nabla' f - \nabla f) , \nabla' h + \nabla h\rangle - \cE(f,h) =
2 \langle R_a \nabla' f ,\nabla' h \rangle
 -2 \langle R_a\nabla f ,\nabla h \rangle,
\ee
and
\be
 \langle R_s (\nabla f - \nabla' f) , \nabla' h + \nabla h\rangle - \cE(f,h) =
2 \langle R_s \nabla f ,\nabla' h \rangle
 -2 \langle R_s\nabla' f ,\nabla h \rangle.
\ee
Plugging these identities into \eqref{PB2LR}, we obtain the  result
\be
\{f,h\}_2^\red =
  \langle \nabla f ,R_a\nabla h \rangle - \langle  \nabla' f ,R_a \nabla' h \rangle
   + \langle \nabla f ,R_s \nabla' h \rangle
 - \langle \nabla' f ,R_s \nabla h \rangle,
\ee
which reproduces the quadratic $r$-matrix bracket \eqref{I3}.

To continue, we evaluate \eqref{T9} for $i=1$.
Substitution  of \eqref{rel**} now gives, at the appropriate arguments,
\be
\langle \nabla_1 F, d_2 H \rangle = \langle (R_a + \frac{1}{2}\id - R_s)[df,L], dh\rangle =
\langle L, [ df, R dh] \rangle - \frac{1}{2} [df, dh]\rangle.
\ee
Here, $R= R_a + R_s$ and we used the standard invariance properties of the trace form \eqref{I1}.
Consequently, we get
\be
\langle \nabla_1 F, d_2 H \rangle - \langle \nabla_1 H, d_2 F \rangle + \langle L, [d_2 F, d_2 H]\rangle
= \langle L, [R df, dh] + [df, R dh]\rangle.
\ee
The right-hand side gives $\{f,h\}_1^\red$, which coincides with the linear $r$-matrix bracket \eqref{I4}.
\end{proof}

\begin{remark}
Let us recall \cite{LP,OR} that $\cG$ carries also a cubic $r$-matrix Poisson bracket which is compatible with the linear and quadratic ones.
It can be obtained from the linear bracket by performing the densely defined change of variables $L \mapsto L^{-1}$, and then extending
the result to the full of $\cG$. For completeness, we  note that the same change of variables is applicable on $T^*G$, too, and the so-obtained
Poisson bracket then leads to the cubic bracket on $\cG$ by the reduction procedure described in the above.
\end{remark}

\section{Discussion}

We explained that the quadratic $r$-matrix bracket \eqref{I3} of the `generalized Toda hierarchy'
\eqref{I6} on $\gl(n,\R)$ is a reduction of a quadratic Poisson bracket  on $T^* \GL(n,\R)$.
This observation escaped previous attention, probably because the convenient form \eqref{PB2} of the relevant parent Poisson bracket
came to light only recently \cite{FeAHP}.
The integrability of the system \eqref{I6} was thoroughly studied
in \cite{DLT} (see also \cite{LP}), together with two other related hierarchies.
These are of the form \eqref{I6}, but instead of $R$ \eqref{I11} use either $R'$ given by
$R'(X) := \frac{1}{2} (X_> + X_0 - X_<)$
or $R'':= R_a$ (which gives the anti-symmetric part of $R'$, too).
We can show that the quadratic $r$-matrix brackets  obtained from \eqref{I3}
by replacing $R$ with $R'$ or $R''$ are  also reductions of the  bracket \eqref{PB2} on $\fM$,
similarly to how the linear $r$-matrix brackets descend \cite{STS2} from \eqref{PB1}.
In the case of $R'$ one may use the group $S':= A' \times B$, where $A'$ is the exponential of the strictly lower
triangular subalgebra of $\cG$. In the case of $R''$ the reduction group is $S''< (G\times G)$ having elements of the form
$(a,b) = (e^{X_0} e^{X_<},  e^{-X_0} e^{X_>})$
which act in the same way as \eqref{T7}.  (The notation refers to \eqref{I7} with arbitrary $X\in \cG$.)
To be precise, in these cases  one  needs to restrict the starting system to $T^* \check G$, where the leading
principal minors  of the elements of $\check G$ are positive, otherwise the reduction procedure is identical to
the presented case, even the crucial equations
\eqref{PB2LR} and \eqref{rel**} keep their form for the corresponding invariant functions.
The open Toda phase space is well known \cite{SurB} to be a Poisson submanifold with respect to the linear $r$-matrix brackets
for any of $R$, $R'$ and $R''$.
However, in contrast to the case of $R$ \eqref{I11}, it is not a Poisson submanifold
with respect to the quadratic brackets associated with $R'$ and $R''$.
It would be interesting to find
the reduction origin of the modified quadratic $r$-matrix brackets of Suris \cite{SurPLA,SurB} that are free from this difficulty.
Another open problem is to extend our treatment
of the quadratic brackets  to spectral parameter dependent $r$-matrices.

 \bigskip
\bigskip
\begin{acknowledgements}
We wish to thank Maxime Fairon for  useful remarks on the manuscript.
This work was supported in part by the NKFIH research grant K134946.
\end{acknowledgements}



\begin{thebibliography}{99}


\addcontentsline{toc}{section}{References}

    \setlength{\parskip}{0em}

\bibitem{DLNT}
P.~Deift, L.-C. Li, T.~Nanda and C.~Tomei,
{\it The Toda flow on a generic orbit is integrable},
Comm. Pure Appl. Math. {\bf 39} (1986) 183–232

\bibitem{DLT}
P.~Deift, L.-C. Li and C.~Tomei,
{\it Matrix factorizations and integrable systems},
Comm. Pure Appl. Math. {\bf 42} (1989) 443-521

\bibitem{FeAHP}
L.~Feh\'er,
{\it Bi-Hamiltonian structure of spin Sutherland models: the holomorphic case},
Ann. Henri Poincar\'e {\bf 22} (2021) 4063–4085;
\href{https://arxiv.org/abs/2101.11484}{\tt arXiv:2101.11484 [math-ph]}


\bibitem{FeNlin}
L.~Feh\'er,
{\it Bi-Hamiltonian structure of Sutherland models coupled to two $\mathfrak{u}(n)^*$-valued spins from Poisson reduction},
to appear in Nonlinearity,
\href{https://arxiv.org/abs/2109.07391}{\tt  arXiv:2109.07391 [math-ph]}


\bibitem{LP}
L.-C.~Li and S.~Parmentier,
{\it Nonlinear Poisson structures and $r$-matrices},
Commun. Math. Phys. {\bf 125} (1989) 545-563


\bibitem{OR}
W.~Oevel and O.~Ragnisco,
{\it $R$-matrices and higher Poisson brackets for integrable systems},
Physica A {\bf 161} (1989) 181-220

\bibitem{STS1}
  M.A.~Semenov-Tian-Shansky,
{\it Dressing transformations and Poisson group actions},
Publ. RIMS {\bf 21} (1985) 1237-1260


\bibitem{STS2}
M.A.~Semenov-Tian-Shansky,
{\it Integrable systems: an $r$-matrix approach},
Kyoto preprint, RIMS-1650, 2008,
www.kurims.kyoto-u.ac.jp/preprint/file/RIMS1650.pdf


\bibitem{SurPLA}
Yu.B.~Suris,
{\it On the bi-Hamiltonian structure of Toda and relativistic Toda lattices},
Phys. Lett. A {\bf 180} (1993) 419-429

\bibitem{SurB}
Yu.B.~Suris,
The Problem of Integrable Discretization:
Hamiltonian Approach,
Birkh\"auser, 2003


\end{thebibliography}
\end{document}